\documentclass[11pt,letterpaper]{article}
\usepackage[margin=1in]{geometry}
\usepackage[T1]{fontenc}
\usepackage{lmodern}
\usepackage{amsmath,amssymb,amsthm,mathtools}
\usepackage{booktabs,array}
\usepackage{microtype}
\usepackage{xspace}
\usepackage[dvipsnames]{xcolor}
\usepackage[colorlinks=true,allcolors=blue]{hyperref}
\usepackage[capitalise,noabbrev]{cleveref}

\newtheorem{theorem}{Theorem}[section]
\newtheorem{lemma}[theorem]{Lemma}
\newtheorem{proposition}[theorem]{Proposition}
\newtheorem{corollary}[theorem]{Corollary}
\theoremstyle{definition}
\newtheorem{definition}[theorem]{Definition}
\theoremstyle{remark}

\newtheorem{example}[theorem]{Example}

\newcommand{\Pclass}{\mathrm{P}}
\newcommand{\NP}{\mathrm{NP}}
\newcommand{\coNP}{\mathrm{coNP}}
\newcommand{\FP}{\mathrm{FP}}
\newcommand{\DeltaTwo}{\Delta_2^{\Pclass}}
\newcommand{\ThetaTwo}{\Theta_2^{\Pclass}}
\newcommand{\SigmaTwo}{\Sigma_2^{\Pclass}}
\newcommand{\SharpP}{\#\Pclass}
\newcommand{\SharpOptP}{\#\!\cdot\!\mathrm{OptP}}
\newcommand{\SharpSAT}{\#\mathrm{SAT}}
\newcommand{\SomePAV}{\textnormal{\textsc{Some-Optimal-PAV-Member}}\xspace}
\newcommand{\AllPAV}{\textnormal{\textsc{All-Optimal-PAV-Member}}\xspace}
\newcommand{\UniquePAV}{\textnormal{\textsc{Unique-PAV}}\xspace}
\newcommand{\CountPAV}{\#\mathrm{OptimalPAV}}
\newcommand{\LexBit}{\textnormal{\textsc{LexMax-SAT-Bit}}\xspace}
\DeclareMathOperator{\pav}{pav}
\DeclareMathOperator{\Opt}{Opt}
\DeclareMathOperator{\val}{val}
\DeclareMathOperator*{\argmax}{arg\,max}
\DeclareMathOperator{\lcm}{lcm}
\newcommand{\bits}[1]{\{0,1\}^{#1}}

\title{The Complexity of Membership, Uniqueness, and Counting\\
for Optimal Proportional Approval Voting Committees}

\author{Yizhou Ai\\University of Toronto\\ \texttt{yizhou.ai@mail.utoronto.ca}}
\date{}

\begin{document}
\maketitle

\begin{abstract}
Proportional Approval Voting (PAV) chooses committees that maximize a sum of harmonic utilities. We study the set of maximizing committees: whether a candidate belongs to some or all of them, whether the optimum is unique, and how many optima exist. When the committee size is part of the input, the three decision problems are $\DeltaTwo=\Pclass^{\NP}$-complete. Uniqueness remains hard for instances with at most two optimal committees. Counting optimal committees is $\SharpOptP$-complete under metric reductions: every function $f$ in this class reduces to an election with exactly $f(x)+1$ optimal committees. The reductions encode satisfying assignments directly as committees and use harmonic marginal rewards to realize binary objectives with polynomially many voters. Each satisfying assignment has a unique committee representation, and fixed clause ballots give these representations the same clause score. We also prove Turing equivalence with $\SharpSAT$ and show that membership of the counting problem in $\SharpP$ would imply $\NP=\coNP$.
\end{abstract}

\section{Introduction}
\label{sec:introduction}

Approval ballots describe which candidates a voter would like to see on a committee. Aggregating these ballots requires deciding how additional representatives for one voter should compare with representation for another. Approval-based multiwinner rules make different choices about this tradeoff~\cite{lackner2023multiwinner}. Proportional Approval Voting (PAV) gives a voter utility $H_r=1+1/2+\cdots+1/r$ from a committee containing $r$ approved candidates, with $H_0=0$, and maximizes the sum of these utilities. Its marginal benefit for a voter decreases from $1$ to $1/2$, $1/3$, and so on as more approved candidates are selected.

PAV has a close connection to proportional representation. It satisfies extended justified representation, a requirement that sufficiently large groups with sufficiently many common approvals receive an appropriate level of representation~\cite{aziz2017jr}. Quantitative approaches, such as proportionality degree, also evaluate committee rules through the satisfaction they guarantee to groups of voters~\cite{skowron2021proportionality}. These perspectives explain why the committees maximizing the PAV objective are of interest beyond their numerical scores.

The optimization rule can, however, return several committees. A score guarantee does not tell us whether a particular candidate occurs in any winning committee, whether the candidate occurs in every winning committee, or whether all winning committees coincide. These properties describe what is determined by the ballots before any tie-breaking convention is applied. The number of maximizing committees further measures the size of the choice left unresolved by the objective.

\begin{example}
\label{ex:election}
Let the candidates be $a,b,c$, let two voters approve $\{a,b\}$ and $\{c\}$, respectively, and select two candidates. All possible committees are listed below.
\[
\begin{array}{c|ccc}
  W & \{a,b\} & \{a,c\} & \{b,c\}\\ \hline
  \pav(W)&3/2&2&2
\end{array}
\]
The candidate $c$ belongs to every optimal committee. Each of $a$ and $b$ belongs to some, but not every, optimal committee. The optimum is not unique, and the number of optimal committees is two. Adding one voter who approves only $a$ changes the scores to $5/2$, $3$, and $2$, respectively, leaving $\{a,c\}$ as the unique optimum.
\end{example}

These questions have prior complexity results: Yang~\cite{yang2023manipulating} proves $\coNP$-hardness of universal candidate membership, and Janeczko and Faliszewski~\cite{janeczko2023ties} prove $\coNP$-hardness of uniqueness and $\SharpP$-hardness of counting PAV winners. We give exact completeness classifications for these questions and for existential candidate membership. The main difficulty is to distinguish all committees at the exact optimum. Although a PAV score is at most $nH_k$ for $n$ voters and committee size $k$, its denominator must also be taken into account. Harmonic increments can distinguish objectives on a much finer scale than a count of satisfied voters. Our reductions use this precision while keeping the numbers of voters and candidates polynomial in the source input.

\subsection{Our Results}

The committee size is part of the input. For every fixed committee size, one can instead enumerate all committees in polynomial time.

\begin{table}[ht]
\centering
\small
\begin{tabular}{@{}p{0.43\linewidth}p{0.25\linewidth}@{}}
\toprule
Problem & Complexity\\
\midrule
\SomePAV & $\DeltaTwo$-complete\\
\AllPAV & $\DeltaTwo$-complete\\
\UniquePAV & $\DeltaTwo$-complete\\
$\CountPAV$ & $\SharpOptP$-complete\\
\bottomrule
\end{tabular}
\label{tab:results}
\end{table}

For candidate membership, both the existential and universal question are $\DeltaTwo$-complete. Uniqueness itself is also $\DeltaTwo$-complete.

For counting, we prove $\SharpOptP$-completeness under metric reductions. For every $f\in\SharpOptP$, we construct an election $\mathcal E_x$ with
\[
    \CountPAV(\mathcal E_x)=f(x)+1.
\]
The additive term accommodates source inputs with no feasible witnesses, since every election has an optimal committee. If the source always has a feasible witness, the construction can preserve its count exactly. We also prove Turing equivalence with $\SharpSAT$ and show that $\CountPAV\in\SharpP$ would imply $\NP=\coNP$.

The main construction represents each variable by a pair of candidates and uses fixed collections of clause voters. These voters give every satisfying assignment the same score, while each unsatisfied clause loses at least one point. Harmonic increments encode a binary objective as a reward with variation less than one, so they order the satisfying assignments without offsetting a clause violation. The Chinese remainder theorem realizes this reward with polynomially many voters. One large parameter forces the auxiliary candidates and one candidate from each variable pair.

The second ingredient controls multiplicity. Variable choices and the forced auxiliary candidates uniquely determine each committee. Together with the standard unique-extension encoding of polynomial-time computations as formulas, this lets the counting reduction preserve every source optimum. A separate representative witness breaks ties only where one representative is needed, leaving the witnesses being counted tied.

\subsection{Related Work}
\label{sec:related}

Aziz et al.~\cite{aziz2015computational} established NP-hardness of computing a PAV winning committee and $\coNP$-completeness of their winning-set verification problem, even when every voter approves two candidates. Their presentation fixes a tie-breaking order. Skowron, Faliszewski, and Lang~\cite{skowron2016collective} obtained a broader hardness result for selection with ordered weighted average utilities, which includes PAV. These results concern computing a committee, meeting a score threshold, or checking a supplied committee. Candidate membership quantifies over all committees attaining the best score, and its classification requires more than hardness of the score-threshold problem.

Yang's extended study of manipulation and control~\cite[Theorem~12]{yang2023manipulating} proves $\coNP$-hardness of deciding whether a designated candidate belongs to every winning committee for a class of Thiele rules including PAV. This is exactly \AllPAV; the result already holds when each ballot has at most two approvals. It arises as the case of constructive control with one designated candidate and no permitted additions. Our $\DeltaTwo$-completeness result strengthens this lower bound for general profiles. For other multiwinner rules, Sonar, Dey, and Misra~\cite{sonar2020verification} distinguish winner verification from candidate membership and prove $\ThetaTwo$-completeness of the latter for Chamberlin--Courant and Monroe rules, with ranking and approval ballots. Their candidate question has the same existential form as \SomePAV, but uses different objectives. Hemaspaandra, Spakowski, and Vogel~\cite{hemaspaandra2005complexity} establish the related $\ThetaTwo$-completeness of Kemeny winner determination.

Janeczko and Faliszewski~\cite{janeczko2023ties} study precisely the uniqueness and winning-committee counting questions considered here. For PAV, they prove $\coNP$-hardness of uniqueness and $\SharpP$-hardness of counting, together with parameterized hardness results. They also rule out polynomially bounded approximation ratios for the count unless $\Pclass=\NP$. Our uniqueness theorem sharpens their lower bound to $\DeltaTwo$-completeness. We strengthen the existing $\SharpP$-hardness result to $\SharpOptP$-completeness under metric reductions.

Dudycz et al.~\cite{dudycz2020tight} give a polynomial-time approximation algorithm for PAV and a broader family of Thiele rules, using linear programming and pipage rounding, with a matching NP-hardness bound on the achievable ratio. These results approximate the maximum harmonic score. The approximation question differs from counting optima: a near-optimal committee need not reveal exact ties, and a multiplicative score guarantee need not resolve candidate membership or uniqueness.

Yang and Wang~\cite{yang2023parameterized} develop a parameterized analysis of PAV winner determination, including fixed-parameter tractability in the number of voters and hardness for other parameters. Lassota and Sornat~\cite{lassota2026structured} give algorithms for Thiele rules on voter-interval profiles, where each candidate's supporters form an interval in a suitable voter order. They also obtain polynomial-time winner determination when each candidate has at most two supporters and fixed-parameter tractability in the optimum score. These findings concern finding an optimum under structural or parameter restrictions. Our classifications concern general profiles and the entire optimal set; transferring an algorithm to these questions requires preserving the relevant membership or counting information.

Representation axioms such as justified representation (JR), extended justified representation (EJR), and proportional justified representation (PJR) require adequate representation of voter groups~\cite{aziz2017jr,fernandez2017pjr}. Every PAV optimum satisfies EJR~\cite{aziz2017jr}, and an EJR committee can be found in polynomial time~\cite{aziz2018complexity}. Satisfying a representation axiom, however, does not establish membership in the set of exact PAV optima. Our questions concern the structure and size of that optimal set.


\section{Preliminaries}
\label{sec:preliminaries}

We use the standard approval-based committee model~\cite{lackner2023multiwinner}. For a positive integer $r$, write $[r]=\{1,\ldots,r\}$.

\begin{definition}
\label{def:election}
An approval election is a triple $\mathcal E=(C,\mathcal B,k)$, where $C=\{c_1,\ldots,c_m\}$ is a set of $m\ge1$ distinct candidates, $\mathcal B=(B_1,\ldots,B_n)$ is a sequence of $n\ge1$ subsets of $C$, and $k\in[m]$ is the committee size. Voter $i\in[n]$ approves the candidates in $B_i$; equal approval sets at different positions represent different voters. A committee is a subset $W\subseteq C$ with $|W|=k$. All three components, including $k$, are part of the input.
\end{definition}

The input lists the candidates and represents the profile by the $n\times m$ approval matrix whose $(i,j)$ entry is $1$ exactly when $c_j\in B_i$. The integer $k$ and any designated candidate index are encoded in binary. A committee is represented by its $m$-bit characteristic vector in the listed candidate order, so each committee has exactly one encoding.

\begin{definition}
\label{def:pav}
Let $\mathcal E=(C,\mathcal B,k)$ be an approval election. For integers $j\ge1$, let $H_j=\sum_{r=1}^j1/r$, and let $H_0=0$. The PAV score of a set $W\subseteq C$ is
\[
    \pav_{\mathcal E}(W)=\sum_{i=1}^n H_{|B_i\cap W|}.
\]
PAV selects the set of committees
\[
    \Opt(\mathcal E)=\argmax_{\substack{W\subseteq C\\|W|=k}}
                              \pav_{\mathcal E}(W).
\]
Thus, $W\in\Opt(\mathcal E)$ if and only if $W$ is a size-$k$ committee and $\pav_{\mathcal E}(W)\ge\pav_{\mathcal E}(U)$ for every size-$k$ committee $U\subseteq C$.
\end{definition}

We omit the subscript on $\pav$ when the election is understood. The set $\Opt(\mathcal E)$ is nonempty because the family of size-$k$ committees is finite and nonempty. All committees attaining the maximum are retained, without a tie-breaking rule.

\begin{definition}
\label{def:problems}
\label{def:some}
The input to \SomePAV is an approval election $\mathcal E=(C,\mathcal B,k)$ and a designated candidate $c\in C$. The answer is yes if and only if
\[
    \exists W\in\Opt(\mathcal E)\quad c\in W.
\]
\end{definition}

\begin{definition}
\label{def:all}
The input to \AllPAV is an approval election $\mathcal E=(C,\mathcal B,k)$ and a designated candidate $c\in C$. The answer is yes if and only if
\[
    \forall W\in\Opt(\mathcal E)\quad c\in W.
\]
\end{definition}

\begin{definition}
\label{def:unique}
The input to \UniquePAV is an approval election $\mathcal E$. The answer is yes if and only if
\[
    |\Opt(\mathcal E)|=1.
\]
\end{definition}

\begin{definition}
\label{def:count}
The function $\CountPAV$ takes an approval election $\mathcal E$ and has value
\[
    \CountPAV(\mathcal E)=|\Opt(\mathcal E)|.
\]
The output is written in binary. Distinct committees are counted once each, irrespective of the order in which their candidates are listed.
\end{definition}

The count is at most $\binom{m}{k}\le2^m$, so its binary representation has at most $m+1$ bits.

The decision and counting complexity notions used below are reviewed in \Cref{app:complexity}. All decision-problem hardness and completeness results use polynomial-time many-one reductions.

For exact score comparisons, set
\[
    D_k=\lcm(1,\ldots,k),\qquad
    S(W)=D_k\pav(W),\qquad
    S^*=\max_{\substack{W\subseteq C\\|W|=k}}S(W).
\]
Both $D_k$ and the integer score $S(W)$ are computable in polynomial time. Since $k\le m$, the bounds $D_k\le k!$ and $0\le S(W)\le nkD_k$ give polynomial bit lengths for these values and the upper bound $nkD_k$.

\section{Complexity of Optimal Membership}
\label{sec:membership-main}
\label{sec:decision}

\begin{theorem}
\label{thm:decision}
The problems \SomePAV and \AllPAV are each $\DeltaTwo$-complete.
\end{theorem}

\begin{proposition}
\label{prop:upper}
\SomePAV and \AllPAV belong to $\DeltaTwo$.
\end{proposition}
\begin{proof}
For an integer threshold $T$, the query
\begin{equation}
\label{eq:threshold}
    \exists W\subseteq C:\quad |W|=k\ \text{and}\ S(W)\ge T
\end{equation}
is in $\NP$, since a guessed committee can be checked in polynomial time. Binary search over $[0,nkD_k]$ therefore computes $S^*$ using $O(\log(nkD_k+1))$ NP queries, which is polynomial in the input length. To decide \SomePAV, ask whether a size-$k$ committee containing $c$ has score $S^*$. To decide \AllPAV, ask whether a size-$k$ committee excluding $c$ has score $S^*$, and negate the answer. Each final query is in $\NP$.
\end{proof}

For the hardness proof, we use lexicographic satisfiability. An ordering $\pi=(x_1,\ldots,x_s)$ of Boolean variables identifies an assignment with a vector $a=(a_1,\ldots,a_s)\in\bits{s}$, where $a_i$ is the value of $x_i$. For distinct assignments $a,b$, write $a>_{\mathrm{lex}}b$ if $a_r=1$ and $b_r=0$ at the least index $r$ for which $a_r\ne b_r$. Thus, $x_1$ is the most significant variable and $1$ is preferred to $0$.

\begin{definition}
\label{def:lexbit}
An instance of \LexBit is a triple $(\varphi,\pi,h)$ with the following components:
\begin{enumerate}
    \item $\pi=(x_1,\ldots,x_s)$ is a list of distinct Boolean variables, with $s\ge1$;
    \item $\varphi$ is a 3-CNF formula whose literals are over these variables;
    \item $h\in[s]$ is the designated position.
\end{enumerate}
Here, a 3-CNF formula is a conjunction of clauses, each containing between one and three literals. Clauses may repeat variables, and the empty conjunction is interpreted as true. The formula is encoded by its explicit list of clauses, the variable order by $\pi$, and the index $h$ in binary.

Let $\mathcal S(\varphi,\pi)=\{a\in\bits{s}:\varphi(a)=1\}$, with coordinates interpreted according to $\pi$. If $\mathcal S(\varphi,\pi)=\varnothing$, the answer is no. Otherwise, let $a^*$ be its unique maximum under $>_{\mathrm{lex}}$. The answer is yes if and only if $a_h^*=1$.
\end{definition}

Throughout the constructions, each formula has an explicitly specified variable set, which may include variables absent from its clauses. We also write $x^*$ for the maximizing assignment, with $x_h^*$ denoting its value on variable $x_h$.

Krentel~\cite[Section~2.2 and Theorem~3.4]{krentel1988optimization} gives a $\DeltaTwo$-complete version of this problem for general Boolean formulas, where the designated variable is last in the order and unsatisfiable formulas are no-instances.

For our reductions, we need satisfiable 3-CNF instances. Given a source formula $F(y_1,\ldots,y_t)$ with this variable order, introduce fresh variables $a,b$ and apply the standard definitional 3-CNF conversion to
\[
    (\neg a\vee F(y))\wedge\bigl(b\leftrightarrow(a\wedge y_t)\bigr).
\]
Order the resulting variables as $a,y_1,\ldots,y_t,b$, followed by all auxiliary variables, and designate $b$. The formula is satisfiable with $a=0$. If $F$ is satisfiable, lexicographic maximization first chooses $a=1$ and then the maximum satisfying $y$ of $F$, so $b$ records the original answer. If $F$ is unsatisfiable, the formula forces $a=b=0$. This polynomial-time conversion preserves the answer and always produces a satisfiable instance of \LexBit. Hence, \LexBit remains $\DeltaTwo$-hard under this restriction.

\subsection{Reduction from \texorpdfstring{\LexBit}{LexMax-SAT-Bit}}
\label{sec:membership-reduction}

We reduce \LexBit to both membership problems. Let $(\varphi,\pi,h)$ be a source instance with $\varphi$ satisfiable, and let $x^*$ be its lexicographically maximum satisfying assignment under $\pi$.

We represent each variable by two candidates, one for true and one for false. The construction ensures that every optimal committee selects exactly one candidate from each pair, and that these choices encode $x^*$. Let $c$ be the candidate representing the true value of $x_h$. If $x_h^*=1$, then $c$ belongs to every optimal committee; if $x_h^*=0$, then $c$ belongs to none of them. Hence, asking whether $c$ belongs to some optimal committee or to every optimal committee gives the same answer as the source problem.

Clause ballots give the same score whenever their clause is satisfied and a lower score otherwise. Once the variable choices encode a satisfying assignment, fractional rewards order the committees by the designated objective bits. For membership these bits are all formula variables; for counting they will be only the source objective's output bits.

To implement this reduction, we associate one committee with each satisfying assignment and make its PAV score an affine function of a designated binary objective. We prove the construction for an arbitrary choice of objective bits so that it can also be used in the counting reduction.

\begin{theorem}
\label{thm:compiler}
Let $\varphi$ be a satisfiable CNF formula with nonempty clauses of size at most three, and let $b_1,\ldots,b_\ell$ be distinct designated variables, where $\ell\ge1$. In polynomial time, one can construct an approval election $\mathcal E_\varphi$ with committee size $K$ and an injective map $x\mapsto W_x$ from satisfying assignments of $\varphi$ to size-$K$ committees such that
\begin{enumerate}
    \item every optimal committee is $W_x$ for some satisfying assignment $x$;
    \item for constants $\Gamma$ and $L>0$ independent of $x$,
    \begin{equation}
    \label{eq:compiler-score}
        \pav(W_x)=\Gamma+\frac{J(x)}{L},
        \qquad J(x)=\sum_{i=1}^{\ell}2^{\ell-i}b_i;
    \end{equation}
    \item for each formula variable $x_j$, a designated candidate $v_j^1$ belongs to $W_x$ exactly when $x_j=1$.
\end{enumerate}
Consequently, optimal committees are in bijection with satisfying assignments maximizing $J$.
\end{theorem}

The clause voters must make all satisfying assignments tie: any extra preference could change which assignments maximize $J$. We give each violated clause a loss of at least one and keep the entire variation in objective rewards below one. Thus, no reward can offset a clause violation, and among satisfying assignments only $J$ determines the score order.

Delete tautological clauses and repeated occurrences of a literal within a clause. Keep the explicitly specified variable set, including variables absent from the remaining clauses. This preserves every satisfying assignment and its objective value. We continue to write $\varphi$ for the resulting formula; each remaining clause has between one and three literals on distinct variables.

Let $x_1,\ldots,x_s$ be the formula variables. For each $x_i$, introduce the candidate pair $V_i=\{v_i^0,v_i^1\}$. Introduce also $d=2\ell-1$ auxiliary candidates and set
\[
    D=\{d_1,\ldots,d_d\},\qquad
    C=D\mathbin{\dot\cup}\bigcup_{i=1}^s V_i,\qquad K=d+s.
\]
The auxiliary candidates will supply the harmonic increments used for the objective. For each assignment $x\in\bits{s}$, define
\[
    W_x=D\cup\{v_i^{x_i}:i\in[s]\}.
\]
We will force every optimum to contain $D$ and exactly one candidate from each $V_i$. Its variable choices then determine the entire committee.

For a clause $(\lambda_1\vee\cdots\vee\lambda_r)$, write $c_h$ and $f_h$ for the candidates that make $\lambda_h$ true and false, respectively. A single ballot $\{c_1,\ldots,c_r\}$ would favor additional true literals even after the clause is satisfied. The following \emph{clause voters} cancel that preference:
\begin{enumerate}
    \item If $r=1$, add one voter approving $\{c_1\}$.
    \item If $r=2$, add two voters approving $\{c_1,c_2\}$ and one approving $\{f_1,f_2\}$.
    \item If $r=3$, add six voters approving $\{c_1,c_2,c_3\}$. For each $1\le h<j\le3$, add one voter approving $\{c_h,f_j\}$ and one approving $\{f_h,c_j\}$. For each $h\in[3]$, add three voters approving only $f_h$.
\end{enumerate}
Let $q$ be the number of true literals in this clause under $x$. The clause voters' contribution to $W_x$ is
\[
\begin{array}{c|c|cccc}
\text{Clause length}&\text{Voters}&q=0&q=1&q=2&q=3\\ \hline
1&1&0&1&\text{--}&\text{--}\\
2&3&3/2&3&3&\text{--}\\
3&21&15&17&17&17
\end{array}
\]
For a unit clause the score is $H_q$, and for a binary clause it is $2H_q+H_{2-q}$. For a ternary clause, there are $q(3-q)$ pairs of literals with different truth values. The two mixed ballots for a pair contribute $3/2$ when its truth values differ and $2$ when they agree. The total score is therefore
\[
    6H_q+6-\frac12q(3-q)+3(3-q),
\]
which gives the last row of the table. Thus, every satisfied clause attains its fixed maximum, irrespective of how many of its literals are true.

Let $t_r$ be the number of $r$-literal clauses, and let $u_r(x)$ count those not satisfied by $x$. The number of clause voters and their score on a satisfying assignment are, respectively,
\[
    N_C=t_1+3t_2+21t_3,\qquad C_\varphi=t_1+3t_2+17t_3.
\]
For every assignment $x$, their total contribution is
\begin{equation}
\label{eq:clause-score}
    S_C(W_x)=C_\varphi-u_1(x)-\frac32u_2(x)-2u_3(x).
\end{equation}

\begin{example}
\label{ex:clause}
For $\varphi=(x_1\vee x_2)$, two clause voters approve $\{v_1^1,v_2^1\}$ and one approves $\{v_1^0,v_2^0\}$. They give score $3/2$ to $W_{00}$ and score $3$ to each of $W_{01}$, $W_{10}$, and $W_{11}$. The three satisfying assignments remain tied until the objective rewards are added.
\end{example}

We next implement the binary objective as a reward with variation less than one. For a designated variable $b_i$, denote its two candidates by $v_{b_i}^0,v_{b_i}^1$. If $q-1$ approved auxiliary candidates are selected, a voter approving those candidates and $v_{b_i}^1$ gains exactly $H_q-H_{q-1}=1/q$ when $b_i$ changes from zero to one.

Directly creating $2^{\ell-i}$ supporters for bit $i$ could require exponentially many voters. Instead, we realize $2^{\ell-i}/L$ as a sum of small-denominator increments and an integer offset. A singleton voter for the zero candidate supplies the offset: choosing the one candidate forgoes that voter's point. All voter counts remain nonnegative. The following arithmetic lemma provides the coefficients and ensures $L\ge2^\ell$, so $0\le J/L<1$.

\begin{example}
\label{ex:crt}
For $\ell=2$, we have $M=4$, $L=12$, and can take $q_1=4$, $q_2=3$. The two reward coefficients satisfy
\[
    \frac{2}{12}=\frac{2}{4}+\frac{2}{3}-1,
    \qquad
    \frac{1}{12}=\frac{3}{4}+\frac{1}{3}-1.
\]
The integer subtraction is implemented by singleton ballots for the zero candidate, rather than by negative voter counts.
In the first identity, two voters provide increments $1/4$, two provide increments $1/3$, and one singleton supporter of the zero candidate supplies the offset. The reward difference between choosing the one and zero candidates is exactly $1/6$. The second identity similarly gives $1/12$. Together they order bit strings by $2b_1+b_2$.
\end{example}

\begin{lemma}
\label{lem:crt}
Set $M=2\ell$ and $L=\lcm(1,\ldots,M)$. One can compute pairwise coprime prime powers $q_1,\ldots,q_r\le M$ with $L=\prod_{j=1}^r q_j$, and nonnegative integers $\alpha_{ij}<q_j$ and $\beta_i<r$, such that
\begin{equation}
\label{eq:crt}
    \sum_{j=1}^r\frac{\alpha_{ij}}{q_j}-\beta_i
        =\frac{2^{\ell-i}}{L}
    \qquad (1\le i\le\ell).
\end{equation}
All computations take polynomial time, $L\ge2^\ell$, and
\begin{equation}
\label{eq:reward-size}
    N_B:=\sum_{i,j}\alpha_{ij}+\sum_i\beta_i
       \le\ell rM\le4\ell^3.
\end{equation}
\end{lemma}
\begin{proof}
First, $\binom{2\ell}{\ell}$ divides $L$. For any prime $p$, writing $\nu_p$ for the exponent of $p$ in an integer, we have
\[
    \nu_p\binom{2\ell}{\ell}
      =\sum_{a\ge1}\left(
         \left\lfloor\frac{2\ell}{p^a}\right\rfloor
          -2\left\lfloor\frac{\ell}{p^a}\right\rfloor\right)
      \le \lfloor\log_p(2\ell)\rfloor=\nu_p(L).
\]
Each summand is zero or one, and it vanishes for $p^a>2\ell$. Moreover,
\[
    \binom{2\ell}{\ell}
      =\prod_{j=1}^{\ell}\frac{\ell+j}{j}\ge2^\ell,
\]
which proves the lower bound on $L$.

For each prime at most $M$, take its largest power at most $M$. These are the $q_j$; they can be found by a sieve and repeated multiplication. Let $w_i=2^{\ell-i}$. Since $\gcd(L/q_j,q_j)=1$, choose the unique $\alpha_{ij}\in\{0,\ldots,q_j-1\}$ satisfying
\[
    \alpha_{ij}(L/q_j)\equiv w_i\pmod{q_j}.
\]
The Chinese remainder theorem implies that
\[
    \beta_i=\frac{\sum_j\alpha_{ij}L/q_j-w_i}{L}
\]
is an integer. Because $0<w_i<L$ and $0\le\sum_j\alpha_{ij}/q_j<r$, this integer lies strictly between $-1$ and $r$, so $0\le\beta_i<r$. This proves \eqref{eq:crt}. For each $i$, the sum of its multiplicities is at most $r(M-1)+(r-1)<rM$, giving \eqref{eq:reward-size}. Finally, $L\le M!$ has polynomial bit length, and gcd and modular inverse computations use polynomial time.
\end{proof}

Use the coefficients from \Cref{lem:crt}. For every $i,j$, add $\alpha_{ij}$ voters approving
\[
    \{v_{b_i}^1,d_1,\ldots,d_{q_j-1}\},
\]
and for each $i$, add $\beta_i$ voters approving only $v_{b_i}^0$. The $d=2\ell-1$ auxiliary candidates suffice because $q_j\le2\ell$. These are the \emph{reward voters}, whose total number is $N_B$. For every assignment $x$, their contribution is
\begin{align}
\label{eq:reward}
    S_B(W_x)
      &=\sum_{i,j}\alpha_{ij}\left(H_{q_j-1}+\frac{b_i}{q_j}\right)
         +\sum_i\beta_i(1-b_i)\notag\\
      &=B_0+\frac{J(x)}L,
\end{align}
where $B_0=\sum_{i,j}\alpha_{ij}H_{q_j-1}+\sum_i\beta_i$ is independent of $x$. In particular, the difference in reward score between any two such committees is strictly less than one.

To enforce the form $W_x$, we use $P$ group voters per variable pair and $P^2$ singleton supporters per auxiliary candidate. Once the auxiliary candidates are fixed, the group voters favor one choice per pair: a first choice earns $P$, whereas a second earns only $P/2$. Choose
\begin{equation}
\label{eq:parameters}
    P=2K(N_C+N_B+1).
\end{equation}
For each variable $x_i$, the \emph{group voters} approve $V_i$; the singleton supporters of $d_j$ approve only $d_j$. This completes the election.

\begin{proof}[Proof of \Cref{thm:compiler}]
Suppose a size-$K$ committee omits an auxiliary candidate $d_j$. It then contains at least $s+1$ variable candidates. Replace any one of them by $d_j$. At most $P$ group voters and $N_C+N_B$ clause and reward voters lose a selected approved candidate, each losing at most one point. The singleton voters for $d_j$ gain $P^2$ points. By the choice of $P$, the net gain is at least
\[
    P^2-P-N_C-N_B>0.
\]
Thus, every optimal committee contains $D$ and exactly $s$ variable candidates.

If these variable candidates do not include exactly one from each pair, some pair contributes two and another contributes none. The group score is then at most $P(s-1/2)$. Replacing the variable choices by one candidate from each pair gains at least
\[
    P/2=K(N_C+N_B+1)>(N_C+N_B)H_K,
\]
which exceeds the maximum possible loss from all clause and reward voters. The auxiliary singleton score stays fixed. Consequently, every optimal committee is $W_x$ for some assignment $x$.

For any assignment $x$, the auxiliary singleton voters contribute $dP^2$ and the group voters contribute $sP$. Adding \eqref{eq:clause-score} and \eqref{eq:reward} gives the full score formula
\begin{equation}
\label{eq:assignment-score}
    \pav(W_x)=\Gamma-u_1(x)-\frac32u_2(x)-2u_3(x)+\frac{J(x)}L,
\end{equation}
where $\Gamma=dP^2+sP+C_\varphi+B_0$. If $x$ fails any clause, its penalty is at least one and $J(x)/L<1$, so $\pav(W_x)<\Gamma$. Since $\varphi$ is satisfiable, some satisfying assignment gives score at least $\Gamma$. Hence, every optimal committee represents a satisfying assignment. For such assignments all three penalties vanish, and \eqref{eq:assignment-score} becomes $\pav(W_x)=\Gamma+J(x)/L$. Precisely the assignments maximizing $J$ therefore give optimal committees. Distinct assignments choose different variable candidates, and $v_i^1\in W_x$ exactly when $x_i=1$, proving the required correspondence.

Finally, the construction has $2s+d$ candidates and
\[
    N_C+N_B+sP+dP^2
\]
voters. The clause voters number at most $21(t_1+t_2+t_3)$, and \Cref{lem:crt} gives $N_B\le4\ell^3$. Since $d=2\ell-1$ and $K=s+d$, the parameter $P$ has polynomial numerical magnitude. Reward ballots approve at most $2\ell$ candidates, and all other ballots approve at most three. Thus, the entire election can be constructed in polynomial time. The large integers $L$ and $2^\ell$ occur only in arithmetic computations, whose bit lengths are polynomial.
\end{proof}

\begin{proof}[Proof of \Cref{thm:decision}]
Take an instance $(\varphi,\pi,h)$ of \LexBit with $\varphi$ satisfiable. Apply \Cref{thm:compiler} with all formula variables as objective bits, in their given order:
\[
    J(x)=\sum_{i=1}^s2^{s-i}x_i.
\]
The unique maximizing assignment is the lexicographically maximum satisfying assignment $x^*$. By the bijection in \Cref{thm:compiler}, the constructed election has exactly one optimal committee $W_{x^*}$. Set $c=v_h^1$. Then
\begin{align*}
    x_h^*=1
    &\iff c\in W_{x^*}\\
    &\iff \exists W\in\Opt(\mathcal E_\varphi):c\in W\\
    &\iff \forall W\in\Opt(\mathcal E_\varphi):c\in W.
\end{align*}
Since the satisfiable restriction of \LexBit is $\DeltaTwo$-hard, this proves $\DeltaTwo$-hardness of both membership problems. Their upper bounds follow from \Cref{prop:upper}.
\end{proof}

Combining \Cref{ex:clause,ex:crt} gives scores $\Gamma+1/12$, $\Gamma+2/12$, and $\Gamma+3/12$ for $W_{01}$, $W_{10}$, and $W_{11}$, respectively. Thus, $W_{11}$ is the unique optimum, and testing $v_h^1\in W_{11}$ recovers the $h$th bit of its assignment.

The two membership problems are not complements: the complement of membership in some optimum requires exclusion from every optimum. They share the reduction above because the constructed instance has a unique optimum, so all optimal committees agree on the designated candidate.

\section{Complexity of Uniqueness}
\label{sec:uniqueness}

\begin{theorem}
\label{thm:uniqueness}
\UniquePAV is $\DeltaTwo$-complete. Hardness holds even when there are at most two optimal committees.
\end{theorem}

For the upper bound, compute $S^*$ as in the proof of \Cref{prop:upper}. Ask whether two distinct size-$k$ committees both have score $S^*$. This is an NP query, and its negation decides uniqueness.

For hardness, we add one variable and one clause to the source formula. The additional variable has one permitted value when the designated bit is one and two permitted values when it is zero. We leave this variable out of the objective.

\begin{proof}[Proof of \Cref{thm:uniqueness}]
Take an instance of \LexBit with satisfiable formula $\varphi$, variable order $x_1,\ldots,x_s$, and designated variable $x_h$. Introduce a fresh variable $z$ and set
\[
    \psi(x,z)=\varphi(x)\wedge(\neg x_h\vee\neg z).
\]
This is a satisfiable 3-CNF formula: every satisfying assignment of $\varphi$ extends by setting $z=0$. Apply \Cref{thm:compiler} to $\psi$, designating only $x_1,\ldots,x_s$ as objective bits, in their given order. The objective is
\[
    J(x,z)=\sum_{i=1}^s2^{s-i}x_i.
\]
Every maximizing assignment restricts to the lexicographically maximum satisfying assignment $x^*$ of $\varphi$. If $x_h^*=1$, the new clause forces $z=0$; if $x_h^*=0$, both values of $z$ satisfy it and give the same objective. The bijection in \Cref{thm:compiler} therefore gives
\begin{equation}
\label{eq:one-or-two}
    |\Opt(\mathcal E_\psi)|=
    \begin{cases}
        1,&x_h^*=1,\\
        2,&x_h^*=0.
    \end{cases}
\end{equation}
This reduces the satisfiable restriction of \LexBit to \UniquePAV, proving $\DeltaTwo$-hardness even for instances with at most two optimal committees. The upper bound was established at the start of this section.
\end{proof}

\begin{example}
Let $\varphi=(x_1\vee x_2)\wedge(\neg x_1\vee\neg x_2)$. Its satisfying assignments are $01$ and $10$, with lexicographic maximum $10$. Adding $(\neg x_1\vee\neg z)$ leaves only the maximizing extension $(x_1,x_2,z)=(1,0,0)$ and hence one optimal committee. Adding $(\neg x_2\vee\neg z)$ instead leaves $(1,0,0)$ and $(1,0,1)$ and hence two optimal committees.
\end{example}

\section{Complexity of Counting Optimal Committees}
\label{sec:counting}

\begin{theorem}
\label{thm:counting}
The function $\CountPAV$ is $\SharpOptP$-complete under metric reductions. More precisely, for every $f\in\SharpOptP$ there is a polynomial-time construction of an approval election $\mathcal E_x$ such that
\begin{equation}
\label{eq:count-main}
    \CountPAV(\mathcal E_x)=f(x)+1.
\end{equation}
If a witness representation of $f$ always has a nonempty feasible set, there is instead a parsimonious reduction from $f$ to $\CountPAV$.
\end{theorem}

The main issue is multiplicity. Lexicographically ordering every variable would isolate one satisfying assignment and destroy the count we want to preserve. We instead encode the source objective as designated formula variables and reward only those bits. There is a separate issue when the source has no feasible witness: a PAV instance cannot have zero optima. The next construction makes the source feasible on every input and adds exactly one to its number of optimal witnesses.

\begin{example}
In \Cref{ex:clause}, the formula $x_1\vee x_2$ has three satisfying assignments. If the source objective is constant, all three must remain optimal. Reusing the membership objective $2x_1+x_2$ would retain only $11$ and change the answer from three to one. For counting, only the bits of the source objective may receive rewards; its witnesses remain tied whenever their objective values agree.
\end{example}

To add one optimal witness, we use a representative $t$ that selects the source optimum value and is itself made unique. A second witness $y$ ranges over all source optima. One extra flag value supplies a single additional solution at the selected value. This separates the choice that must be unique from the choices that must be counted.

\begin{example}
\label{ex:plusone}
Suppose the feasible source strings are $00,01,10$, with values $3,3,1$. There are two source optima. Choose the larger optimal representative $t=01$, while allowing $y$ to be either $00$ or $01$. The intended optimal tuples $(a,\sigma,t,y)$ are
\[
\begin{array}{c|c|c|c|l}
  a&\sigma&t&y&\text{role}\\ \hline
  1&0&01&00&\text{one additional optimum}\\
  1&1&01&00&\text{source optimum }00\\
  1&1&01&01&\text{source optimum }01
\end{array}
\]
The flag $\sigma$ keeps the additional tuple distinct even though $00$ is also a source optimum. A separate tuple with $a=0$ is a fallback: it loses to every active tuple when the source is feasible, and is the only feasible tuple when the source has no feasible string. The construction below realizes exactly this behavior for arbitrary sources.
\end{example}

\begin{lemma}
\label{lem:plusone}
For every $f\in\SharpOptP$, one can construct polynomial-time feasibility and objective computations defining a nonempty feasible set whose number of optimal witnesses on input $x$ is exactly $f(x)+1$.
\end{lemma}
\begin{proof}
Use $A(x,y)$, $v(x,y)$, and $p$ from \Cref{def:sharpopt}, and abbreviate $p=p(|x|)$. We may assume $p\ge1$ by adding a bit fixed to zero if necessary. The new witnesses are tuples
\[
    (a,\sigma,t,y),\qquad a,\sigma\in\{0,1\},\quad t,y\in\bits p.
\]
Their feasibility conditions are as follows:
\begin{enumerate}
    \item If $a=0$, require $\sigma=0$ and $t=y=0^p$.
    \item If $a=1$, require $A(x,t)$ and one of
    \[
        \sigma=0,\quad y=0^p,
        \qquad\text{or}\qquad
        \sigma=1,\quad A(x,y),\quad v(x,y)=v(x,t).
    \]
\end{enumerate}
The first branch supplies exactly one feasible witness. In the second branch, $t$ will select one representative of the optimal source value, while $y$ will retain all source witnesses at that value. Define
\begin{equation}
\label{eq:plusone-objective}
    J(a,\sigma,t,y)=
    \begin{cases}
        0,&a=0,\\
        1+2^p v(x,t)+\val(t),&a=1,
    \end{cases}
\end{equation}
where $\val(t)$ is the binary integer represented by $t$. This objective and the feasibility predicate are computable in polynomial time. If $v$ uses at most $B$ bits, then $J\le2^{p+B}$, so $p+B+1$ bits suffice for its output.

If $Y_x$ is empty, no witness with $a=1$ is feasible. The single witness with $a=0$ is optimal, and the number of optima is $1=f(x)+1$.

Suppose $Y_x$ is nonempty. Every witness with $a=1$ has objective at least one and therefore defeats the witness with $a=0$. A unit increase in $v(x,t)$ increases the term $2^pv(x,t)$ by more than the largest possible decrease in $\val(t)$. Maximizing \eqref{eq:plusone-objective} thus first maximizes $v(x,t)$, and then chooses the unique lexicographically largest optimal representative $t^*$. With $t=t^*$ fixed, the optimal witnesses consist of one tuple with $\sigma=0,y=0^p$, and one tuple with $\sigma=1$ for each optimal source witness $y$. The bit $\sigma$ keeps the additional tuple distinct even if $0^p$ itself is an optimal source witness. There are exactly $f(x)+1$ optimal tuples.
\end{proof}

Only the representative $t$ receives a lexicographic preference in this construction. The objective does not distinguish the optimal choices of $y$. We can now combine \Cref{lem:plusone} with \Cref{thm:compiler} to prove the counting theorem.

\begin{table}[ht]
\centering
\small
\begin{tabular}{@{}p{0.43\linewidth}p{0.5\linewidth}@{}}
\toprule
Transformation & Reason no multiplicity is introduced\\
\midrule
Feasible tuple to satisfying assignment & The standard formula encoding has one satisfying extension per feasible input.\\
Satisfying assignment to committee & Variable choices determine the committee; all auxiliary candidates are fixed.\\
\bottomrule
\end{tabular}
\caption{The two bijections used in the counting reduction. The additional optimum is introduced before these transformations.}
\label{tab:bijections}
\end{table}

\begin{proof}[Proof of \Cref{thm:counting}]
For membership, guess a candidate characteristic vector, accept exactly when it encodes a size-$k$ committee $W$, and output its integer score $S(W)$. Each committee has one accepting path. The paths with maximum output are therefore counted by $\CountPAV$. This proves membership in $\SharpOptP$ as characterized in \Cref{def:sharpopt}.

For hardness, fix any $f\in\SharpOptP$ and an input $x$. Apply \Cref{lem:plusone}. Using the standard 3-CNF encoding of polynomial-time computations by full gate equivalences, construct in polynomial time from $x$ a formula $\varphi_x$ with free inputs exactly $(a,\sigma,t,y)$. Each feasible tuple has exactly one satisfying extension, and infeasible tuples have none. Distinct designated variables $b_1,\ldots,b_\ell$ record $J$ with a fixed sufficient length $\ell\ge1$, most significant bit first; constants and padding are fixed.

The formula is satisfiable because \Cref{lem:plusone} supplies a feasible tuple. The correspondence preserves
\[
    J=\sum_{i=1}^{\ell}2^{\ell-i}b_i.
\]

Apply \Cref{thm:compiler}, designating only the objective output bits $b_1,\ldots,b_\ell$. No other formula variable receives a reward. There is then a chain of bijections
\begin{align*}
    \{\text{optimal PAV committees}\}
    &\longleftrightarrow
      \{\text{satisfying assignments of $\varphi_x$ maximizing $J$}\}\\
    &\longleftrightarrow
      \{\text{optimal feasible tuples of \Cref{lem:plusone}}\}.
\end{align*}
The first bijection follows from \Cref{thm:compiler}; the second follows from the unique extension of each feasible tuple. By \Cref{lem:plusone}, the resulting number of committees is $f(x)+1$, proving \eqref{eq:count-main}. Every stage is polynomial in $|x|$, including the explicit list of voters. The total postprocessing function $h(x,z)=\max\{z-1,0\}$ recovers $f(x)$, since the constructed oracle value is at least one. This gives a metric reduction.

If the source representation has a feasible witness on every input, omit \Cref{lem:plusone}. Apply the same standard encoding directly to $A(x,y)$ and $v(x,y)$, and then apply \Cref{thm:compiler} with only the output bits of $v$ designated. The same bijection now preserves exactly $f(x)$ optimal witnesses, proving the parsimonious statement.
\end{proof}

Every PAV instance has at least one optimum, so a parsimonious reduction cannot represent a source value of zero. The metric reduction handles this obstruction uniformly through the additive one. We next relate the counting problem to $\SharpSAT$ and other counting classes.

\begin{corollary}
\label{cor:turing}
We have
\[
    \CountPAV\equiv_T^p\SharpSAT,
    \qquad
    \CountPAV\in\FP^{\SharpP}\cap\SharpP^{\NP}.
\]
\end{corollary}
\begin{proof}
The Turing lower bound follows already from the $\SharpP$-hardness of Janeczko and Faliszewski~\cite{janeczko2023ties}; it also follows from \Cref{thm:counting}.

For the upper bound, replace the NP queries used to compute $S^*$ in the proof of \Cref{prop:upper} by $\SharpSAT$ queries. Encode the polynomial-time test $|W|=k$ and $S(W)\ge T$ as a 3-CNF formula with one satisfying extension per accepted candidate characteristic vector. Its count is positive exactly when the threshold is feasible. After binary search determines $S^*$, use the same encoding for $|W|=k$ and $S(W)=S^*$. One final $\SharpSAT$ query then returns exactly the number of optimal committees. This proves the Turing upper bound and membership in $\FP^{\SharpP}$.

For membership in $\SharpP^{\NP}$, guess one committee characteristic vector and reject if its size is not $k$. Query whether a size-$k$ committee has strictly greater integer score. Accept exactly when the answer is no. The accepting paths correspond one-to-one to optimal committees.
\end{proof}

\begin{corollary}
\label{cor:not-sharpp}
If $\CountPAV\in\SharpP$, then $\NP=\coNP$.
\end{corollary}
\begin{proof}
Suppose a nondeterministic polynomial-time machine computes $\CountPAV$ on every instance. Let $L_{\ge2}$ be the language of elections with at least two optimal committees. This language is in $\NP$: guess two distinct accepting paths of the machine and verify both. Their distinction can be checked from their nondeterministic choice sequences, with unused padding fixed to zero.

To prove $\DeltaTwo$-hardness of $L_{\ge2}$, take any language $L_0\in\DeltaTwo$. Closure under complement gives $\overline{L_0}\in\DeltaTwo$, so apply the uniqueness hardness reduction to $\overline{L_0}$. Its output is an election $\mathcal E_x$, with
\[
    |\Opt(\mathcal E_x)|=1
    \iff x\in\overline{L_0}
    \iff x\notin L_0.
\]
Since every election has at least one optimum, $\mathcal E_x\in L_{\ge2}$ if and only if $x\in L_0$. Hence, $\DeltaTwo\subseteq\NP$. As $\coNP\subseteq\DeltaTwo$, we obtain $\coNP\subseteq\NP$, and taking complements gives $\NP=\coNP$.
\end{proof}

\Cref{cor:turing} establishes Turing equivalence with $\SharpSAT$ and an upper bound in $\SharpP^{\NP}$. It does not assert $\SharpP^{\NP}$-hardness. Similarly, \Cref{cor:not-sharpp} is a conditional implication rather than an unconditional separation of counting classes.

\section{Conclusion}
\label{sec:conclusion}

Membership in some or all optimal PAV committees and uniqueness of the optimum are $\DeltaTwo$-complete. Counting the optimal committees is $\SharpOptP$-complete under metric reductions, and the reduction adds exactly one to the source count. These statements concern the whole maximizing set, including ties that would be hidden by selecting one output committee.

The proofs separate clause satisfaction from objective comparison. Fixed clause ballots give every satisfying assignment the same base score, and harmonic increments order these assignments with polynomially many voters. Variable pairs and the standard unique-extension formula encoding preserve the number of optimal solutions. Existing algorithms exploit the number of voters and restrictions on approval profiles~\cite{yang2023parameterized,lassota2026structured}. It remains useful to determine which such restrictions also simplify candidate membership, uniqueness, and counting, and how limits on ballot size constrain the harmonic increments available to hardness constructions.

\bibliographystyle{alphaurl}
\bibliography{reference}

\clearpage
\appendix
\section{Complexity Background}
\label{app:complexity}
\label{sec:complexity}

This appendix reviews the background for our decision and counting classifications. It introduces oracle computation and $\DeltaTwo$, then the class $\SharpOptP$ and the reductions used for counting. Readers familiar with these notions may proceed directly to the proofs.

Write $\Sigma=\{0,1\}$. A decision problem is identified with a language $L\subseteq\Sigma^*$ of yes instances, and its complement is $\overline L=\Sigma^*\setminus L$. A polynomial-time many-one reduction from $L$ to $K$ is a polynomial-time computable map $r:\Sigma^*\to\Sigma^*$ satisfying
\[
    x\in L\quad\Longleftrightarrow\quad r(x)\in K.
\]
For a decision class $\mathcal C$, a language is $\mathcal C$-hard if every language in $\mathcal C$ has such a reduction to it. It is $\mathcal C$-complete if it is also in $\mathcal C$.

The class $\Pclass$ contains the languages decidable in deterministic polynomial time. A language $L$ is in $\NP$ if there are a polynomial $p$ and a polynomial-time predicate $R$ such that
\[
    x\in L\quad\Longleftrightarrow\quad
    \exists y\in\bits{p(|x|)}:\ R(x,y)=1.
\]
Thus, membership has a certificate of polynomial length that can be checked in polynomial time. A language belongs to $\coNP$ when its complement belongs to $\NP$.

An NP oracle answers membership queries to a fixed NP-complete language, such as satisfiability. An oracle machine receives each yes-or-no answer in one step, but must write its queries within its running-time bound. Queries are adaptive when a later query may depend on earlier answers.

\begin{definition}
\label{def:delta}
A language $L$ belongs to $\DeltaTwo=\Pclass^{\NP}$ if a deterministic machine with an NP oracle decides whether $x\in L$ in time polynomial in $|x|$. The machine may choose its queries adaptively and may make polynomially many queries.
\end{definition}

The related-work discussion also uses $\ThetaTwo$ and $\SigmaTwo$. Restricting a deterministic polynomial-time machine to $O(\log(|x|+2))$ adaptive NP queries gives $\ThetaTwo$; equivalently, this class permits polynomially many queries chosen before any answer is received. Allowing nondeterministic polynomial-time computation with an NP oracle gives $\SigmaTwo=\NP^{\NP}$.

Let $\mathbb N=\{0,1,\ldots\}$, with integers written in binary. A function belongs to $\SharpP$ if its value is the number of accepting paths of a nondeterministic polynomial-time machine. Equivalently, it counts fixed-length certificates satisfying a polynomial-time predicate. For example, $\SharpSAT$ counts the satisfying assignments of a Boolean formula.

Optimization adds an objective to the feasible computations. Krentel's class $\mathrm{OptP}$~\cite{krentel1988optimization} takes the maximum nonnegative integer output over the paths of a nondeterministic polynomial-time machine, all of which accept. Hermann and Pichler's class $\SharpOptP$~\cite{hermann2008counting} instead counts accepting paths attaining the maximum output, with count zero when there are no accepting paths. We use the following equivalent formulation in terms of witnesses.

\begin{samepage}
\begin{definition}
\label{def:sharpopt}
A function $f:\Sigma^*\to\mathbb N$ belongs to $\SharpOptP$ if there are a polynomial $p$, a polynomial-time feasibility predicate $A(x,y)$, and a polynomial-time computable nonnegative integer objective $v(x,y)$ of polynomial bit length such that, for
\[
    Y_x=\{y\in\bits{p(|x|)}:A(x,y)\},
\]
we have
\[
    f(x)=
    \begin{cases}
        |\argmax_{y\in Y_x}v(x,y)|,&Y_x\ne\varnothing,\\
        0,&Y_x=\varnothing.
    \end{cases}
\]
\end{definition}
\end{samepage}

To obtain fixed-length witnesses from machine paths, record the nondeterministic choices and fix all unused padding bits to zero. Each accepting path then has one witness. Conversely, a machine guesses a witness, checks $A(x,y)$, and outputs $v(x,y)$ on acceptance. These transformations preserve the objective and the number of witnesses at each value. In particular, a constant objective recovers ordinary $\SharpP$ counting, so $\SharpP\subseteq\SharpOptP$.

We also use oracle versions of counting and function classes. The notation $\FP^{\NP}$ denotes functions computed in deterministic polynomial time with an NP oracle. In $\FP^{\SharpP}$, the oracle returns the exact value of a $\SharpP$ function. The class $\SharpP^{\NP}$ counts accepting paths of nondeterministic polynomial-time machines that can query an NP oracle. Thus, an oracle algorithm for a counting function and a machine whose accepting paths are being counted describe different resources.

For counting functions, the reduction specifies how the target count may be used to compute the source count.

\begin{definition}
\label{def:metric}
A polynomial-time metric reduction~\cite{krentel1988optimization} from $f:\Sigma^*\to\mathbb N$ to $g:\Sigma^*\to\mathbb N$ consists of total polynomial-time computable functions $r:\Sigma^*\to\Sigma^*$ and $h:\Sigma^*\times\mathbb N\to\mathbb N$ such that
\[
    f(x)=h(x,g(r(x))).
\]
\end{definition}

A metric reduction uses one value of $g$ and computes the answer from that value and the original input. A parsimonious reduction requires $f(x)=g(r(x))$, preserving the count exactly.

We write $f\le_T^p g$ if a deterministic polynomial-time algorithm computes $f$ using an oracle for $g$. It may make polynomially many adaptive queries. The functions are polynomial-time Turing equivalent, written $f\equiv_T^p g$, when both $f\le_T^p g$ and $g\le_T^p f$ hold.

For a counting class, hardness requires a reduction from every function in the class, and completeness additionally requires membership. The reduction notion is specified in the completeness statement.

\end{document}